\documentclass[amsart,11pt,oneside,english]{article}
\usepackage{amssymb,amsmath,amsfonts, amsmath,mathtools,
amsthm,euscript,mathrsfs,MnSymbol,verbatim,enumerate,multirow,bbding,slashed,multicol,color,array, esint,babel, tikz,tikz-cd,tikz-3dplot,tkz-graph,pgfplots,ytableau,graphicx,float,rotating,hyperref,geometry,mathdots,savesym,cite, wasysym,amscd,graphicx,pifont,float,setspace,multirow,wrapfig,picture,subfigure, amsthm, enumitem}
\usepackage[normalem]{ulem} 
\usepackage[utf8]{inputenc}
\usepackage{prettyref}
\usepackage[T1]{fontenc}
\usepackage{datetime}
\usepackage{thm-restate}

\numberwithin{equation}{section}
\usetikzlibrary{arrows,positioning,decorations.pathmorphing,   decorations.markings, matrix, patterns}
\hypersetup{colorlinks=true}
\hypersetup{linkcolor=black}
\hypersetup{citecolor=black}
\hypersetup{urlcolor=black}
\theoremstyle{plain}
\newtheorem*{thm*}{Theorem}
\theoremstyle{plain}%
\newtheorem{thm}{Theorem}[section]

\newtheorem{cor}[thm]{Corrollary}

\theoremstyle{definition}
\newtheorem{defn}[thm]{Definition}

\newtheorem{rem}[thm]{Remark}

\numberwithin{equation}{section}

\tikzset{
  big arrow/.style={
    decoration={markings,mark=at position 1 with {\arrow[scale=1.5,#1]{>}}},
    postaction={decorate},
    shorten >=0.4pt},
  big arrow/.default=black}

\makeatother

\begin{document}

\begin{titlepage}
\begin{center}
\vspace{2cm}
{\Huge\bfseries The suspended pinch point\\
  and SU($2$)$\times$U($1$) gauge theories  \\  }
\vspace{2cm}
{\Large
Mboyo Esole$^{\diamondsuit}$ and Sabrina Pasterski$^\clubsuit$\\}
\vspace{.6cm}
{\large $^{\diamondsuit}$ Department of Mathematics, Northeastern University}\par
{\large  360 Huntington Avenue, Boston, MA 02115, USA}\par
\vspace{.3cm}
{{\large $^\clubsuit$ Jefferson Laboratory,  Harvard University}\par
{ 17 Oxford Street, Cambridge, MA 02138, USA}\par}
\vspace{2cm}
{ \bf{Abstract}}\\
\end{center}
We show that the suspended pinch point can be seen as an elliptically fibered variety with singular fibers of type I$_2$ over codimension-one points of the base and a torsionless Mordell--Weil group of rank one. In the F-theory algorithm, this corresponds to a Lie group 
$SU(2)\times U(1)$. We also identify the matter content as given by the  direct sum of the adjoint representation (with zero U($1$)-charge) and the fundamental representation with U($1$)-charge $\pm 1$.
We then study the geometry of an SU($2$)$\times$U($1$)-model given by a compact elliptically fibered variety with the singularities of a suspended pinch point. 
We describe in detail the  crepant resolutions and the network of flops  of this geometry.
We compute topological invariants including the Euler characteristic and Hodge numbers.
We also study the weak coupling limit of this geometry and show that it corresponds to an orientifold theory with an Sp($1$)-stack transverse to the orientifold and two  brane-image-branes wrapping the orientifold.

\vfill 

\end{titlepage}

\pagestyle{empty} 
\tableofcontents

\newpage 
\pagestyle{plain}
\setcounter{page}{1}

\section{Introduction}

Elliptic fibrations are elegant algebraic varieties whose geometries can surprisingly encode essential information about the structure of gauge theories \cite{Bershadsky:1996nh,Morrison:1996pp, Anderson:2017zfm}.   Insights from the physics of gauge theories and string dualities  shed light on different aspects of the geometry of elliptic fibrations, creating opportunities for new results in both mathematics and physics.

In mathematics, elliptic fibrations naturally appear in the classification of algebraic surfaces, which was the motivation for Kodaira's seminal work on elliptic surfaces \cite{Kodaira}. They are also central in number theory, where the Riemann Roch theorem plus base change gives a Weierstrass model for any curve of genus one with a rational point. The singular fibers of a Weierstrass model were  classified by  N\'eron \cite{Neron}, and  Tate's algorithm gives simple rules for determining the type of a singular fiber by analyzing the valuations of its coefficients, its $j$-invariant, and its discriminant  \cite{Tate.Alg,Papadopoulos}. 
Elliptic fibrations  also appear naturally in the classification of higher dimensional varieties as the output of Mori's minimal model program since they correspond to a special case of {\em Calabi-Yau fibered spaces} with  one-dimensional fibers \cite{Kawamata.97} and 
any minimal model satisfying the abundance theorem yields a Calabi-Yau fibered space \cite{Kawamata.97}. 
There are examples of (singular) elliptic fibrations which are Mori dream spaces \cite{Garbagnati,Artie}.
Elliptic fibrations are also instrumental in the study of Calabi-Yau varieties and mirror symmetry. 
Several surveys indicate that  a large majority of known Calabi-Yau threefolds are elliptically fibered. 
Besides, they have rich a structure of flops that can be discussed in a context more general than toric geometry.

Over the years,  elliptic fibrations have become precision tools in the art of geometrically engineering gauge theories in M-theory and F-theory compactifications (for a review of F-theory, see \cite{Denef.F,Weigand:2018rez,Park}).  As a testimony to their ubiquitous influence on our current understanding of quantum field theory, elliptic fibrations provide an organizational tool for the classification of six-dimensional superconformal field theories \cite{6DCFTs}. In five-dimensional gauge theories obtained by compactification of M-theory on elliptically fibered Calabi-Yau threefolds, distinct minimal models can be identified with distinct chambers of the Coulomb phases of an associated five-dimensional gauge theory obtained by a compactification of M-theory on the Calabi-Yau threefold \cite{IMS, Witten}.  Crepant birational transformations between elliptic fibrations  model physical phase transitions between different Coulomb phases of the gauge theory \cite{IMS,Witten}. 
Certain degenerations of elliptic fibrations yield weak coupling limits of F-theory \cite{Sen:1997gv,CDE,AE1,AE2,Esole:2012tf}  that have provided unexpected mathematical relations inspired by the cancellation of tadpoles \cite{AE1,AE2,EFY,EKY1,CDE,Clingher:2012rg} and connection to K-theory \cite{CDE}.  The study of the Coulomb phases of five-dimensional gauge theories \cite{IMS,Witten}  indicates that  the geography of flops of an elliptic fibration can be modeled by a hyperplane arrangement \cite{ESY1,ESY2,EJJN1, EJJN2}. 
In this context, there  has been real progress recently toward the classification of flops associated to gauge theories of small ranks both for simple groups  \cite{ESY1,ESY2,G2,F4,SO,E7,ES}, and semi-simple groups  \cite{SU2SU3,SO4,SU2G2,EKY2}. 
Their Euler characteristic and other characteristic numbers are also well understood \cite{Euler,Char1, Char2}.

The {\em Mordell--Weil group} of an elliptic fibration is the finitely generated Abelian group of its rational sections.   The  {\em  Mordell--Weil rank} of an elliptic fibration is the rank of its Mordelll--Weil group.
The Mordell--Weil group of an elliptic fibration is a birational invariant. 
 A non-trivial Mordell--Weil rank  has a dramatic impact on the singularities of the Weierstrass model of the elliptic fibration.  One reason for this is that they are not $\mathbb{Q}$-factorial in contrast to Weierstrass models with Mordell--Weil group of rank zero. 
A  Weierstrass model over a smooth base is $\mathbb{Q}$-factorial if and only if its Mordell--Weil rank is zero \cite[Theorem 8.3]{Jagadeesan}. 
An elliptic fibration with zero Mordell--Weil rank is said to be {\em extremal}. 
In the F-theory algorithm, the rank of the Abelian sector of the reductive Lie group associated with an elliptic fibration is given by its  Mordell--Weil rank.
The role of the Mordell--Weil group in F-theory is studied in \cite{Morrison:2012ei,Mayrhofer:2014opa}. We refer to  \cite{Weigand:2018rez,Cvetic:2018bni} for a pedagogical review.

The study of flops between crepant resolutions of Weierstrass models with zero Mordell--Weil rank but  non-trivial  Mordell--Weil torsion was started in \cite{SO,EKY2}. The study of   flops of $G$-models with a nonzero Mordell--Weil rank is still in its infancy and  in need of  much more attention. 
 The simplest case to consider would be an elliptic fibration with a torsionless Mordell--Weil group of rank one  where the  associated Lie algebra $\mathfrak{g}$ is   of  type A$_1$. Such a model is a $G$-model and
 has minimal positive rank both for the Abelian and the semi-simple sector as its corresponds to 
  $$
G=\text{ SU($2$)$\times$U($1$)}.
$$

\subsection{An  \text{ SU($2$)$\times$U($1$)}-model  hiding in plain sight}
The  aim of this paper is to discuss the geometry of an elliptic fibration that is hiding in plain sight 
 and  corresponds to a simple   \text{ SU($2$)$\times$U($1$)}-model. 
The geometry of interest is  the  suspended pinch point defined by the following binomial equation in $\mathbb{C}^4$ \cite{Morrison:1998cs}:
\begin{equation}\label{eq:xiSPP}
Z_0: \quad x_0 x_1 -x_2 x_3^2=0.
\end{equation}
Upon performing the following  linear change of variables, 
\begin{equation}
(x_0,x_1, x_2, x_3)\to (y-s, y+s,x+t,x)
\end{equation}
 the suspended pinch point  $Z_0$ reveals itself as the following Weierstrass model 
\begin{equation}\label{eq:SSPWM}
Y_0:\quad y^2=x^3+tx^2+s^2,
\end{equation}
with discriminant and $j$-invariant
\begin{equation}
\Delta= s^2(4t^3+27 s^2), \quad \quad j=-\frac{2^8t^6}{s^2(4t^3+27 s^2)}.
\end{equation}
The elliptic fibration $Y_0$ defined by the Weierstrass model in equation \eqref{eq:SSPWM} is a special case of an   SU($2$)$\times$U($1$)-model.
We call it an {\em SPP elliptic fibration} as it has the singularity of a suspended pinch point. 
This is an elliptic fibration with a discriminant locus containing an irreducible component  -- the smooth divisor $S=V(s)$ corresponding to the vanishing locus $s=0$ -- whose generic point has a singular fiber with dual graph $\widetilde{\text{A}}_1$.   The  Mordell--Weil group has rank one and is generated by the regular sections 
\begin{equation}
\Sigma^\pm: \quad x=y\pm s=0,
\end{equation}
that are opposite of each other with respect to the Modell--Weil group law.

When considering the Weierstrass model of equation \eqref{eq:SSPWM} as a new model for the SU($2$)$\times$U($1$)-model,  
the base $B$ could be of arbitrary dimension. If we denote by $\mathscr{L}$ the fundamental line bundle of the Weierstrass model, the discriminant locus is a section of $\mathscr{L}^{\otimes 12}$, thus 
the divisor $S=V(s)$ supporting the fiber I$_2^{\text{ns}}$ is necessarily a section of $\mathscr{L}^{\otimes 3}$ and the divisor  $T=V(t)$ is a zero section of  $\mathscr{L}^{\otimes 2}$.  In the Calabi-Yau case, $c_1(\mathscr{L})=-K$ where $K$ is the canonical class of the base of the fibration. 
In particular,  if the base is a surface, $S$ is  a curve of genus  
\begin{equation}
g(S)=1+6K^2.
\end{equation}
 In particular, as for the SO($6$), SO($5$), and SO($3$)-models \cite{SO},  the divisor $S$ can never be a rational curve. Thus this model will always have matter in the adjoint representation in a compactification of M-theory to a five-dimensional gauge theory or  F-theory to  a six-dimensional gauge theory. 

\subsection{Minimal models, hyperplane arrangement, and flops}

A compactification of M-theory on a Calalabi--Yau threefold  gives a five-dimensional supergravity theory  with eight supercharges (${\cal N}=1$) \cite{Cadavid:1995bk,Ferrara:1996wv}.  In the Coulomb regime, there may be multiple phases separated by hyperplanes on which massive hypermultiplets become massless \cite{IMS,Witten}. The chamber structure of the Coulomb phases is due to the presence of absolute values that appear in the one-loop quantum correction to the prepotential that controls the dynamics of gauge fields in an ${\cal N}=1$ 5D gauge theory \cite{IMS}. It is then natural to introduce an explicit hyperplane arrangement \cite{ESY1,ESY2,EJJN1,EJJN2} to study the structure of the Coulomb phases.  

\begin{defn}[{The hyperplane arrangement  I($\mathfrak{g},\mathbf{R})$}]\label{Def:IgR}
Let $\mathfrak{g}$ be a semi-simple Lie algebra and $\mathbf{R}$ a representation of $\mathfrak{g}$. 
The hyperplane arrangement  I($\mathfrak{g},\mathbf{R})$  is defined inside the dual fundamental Weyl chamber of the Lie algebra  $\mathfrak{g}$. Its walls are the kernels of the weights of the representation $\mathbf{R}$. 
\end{defn}
 
\noindent The chambers  of  I$(\mathfrak{g}, \mathbf{R})$  correspond to the Coulomb phases of  an ${\cal N}=1$ 5D gauge theory with gauge algebra $\mathfrak{g}$ and matter transforming in the representation $\mathbf{R}$ of $\mathfrak{g}$ \cite{IMS,ESY1,ESY2,Hayashi:2014kca}.

Crepant resolutions of a Weierstrass model are minimal models over the Weierstrass model. 
We construct explicitly  the  minimal models over the Weierstrass model of equation \eqref{eq:SSPWM} by giving explicit crepant resolutions. 
In F-theory, the duality with five-dimensional gauge theories with eight supercharges provides a one-to-one correspondence between the Coulomb phases of a gauge theory with gauge algebra $\mathfrak{g}$ and hypermultiplets transforming in representation $\mathbf{R}$ of $\mathfrak{g}$ and the minimal models of a Weierstrass model corresponding to a  $G$-model with an associated representation $\mathbf{R}$. 
We show that the  representation associated to the SU($2$)$\times$U($1$)-model considered in this paper is\footnote{
We recall that any irreducible representation $\mathbf{r} $ of a product $H_1\times H_2$ of two compact Lie groups $H_i$ ($i=1,2$)  is the tensor product $\bf{r}=\bf{r_1}\otimes\bf{r_2}$ where $\bf{r_i}$ is an irreducible representation of $H_i$ ($i=1,2$). 
We denote a representation 
of $\bf{n}\otimes\bf{Y}$ of 
of SU($2$)$\times$U($1$) as 
$\bf{ n_{Y}}$,   where $\bf{n}$ is a representation of SU($2$) and $\bf{Y}$ is a representation of  U($1$). 
The adjoint of SU($2$) is denoted $\bf{3}$ and its fundamental representation is $\bf{2}$. 
We also note that $\bf{2}_{-m}$ should be thought as the CPT conjugate of $\bf{2}_{m}$ since the fundamental of SU($2$) is a pseudo-real representation. 
In the electroweak theory, $\bf{Y}$ is the hypercharge. For example, left-handed leptons  transform in the representation the $\bf{2_{-1}}$. 
} 
\begin{equation}
 \mathbf{R}=\bf{3}_0\oplus \bf{2}_{-1},
\end{equation}
where $\bf{3}_0$ is the adjoint representation of  SU($2$) with zero U($1$)-charge, and $\bf{2}_{-1}$ is the fundamental representation of SU($2$) with U($1$)-charge  $-1/2$. 
The U($1$)-charge can be normalized to $-1$.

This SU($2$)$\times$U($1$)-model provides a unique opportunity to  have a hard check, in the case of a Mordell-Weil group with  a nonzero rank,  on the conjecture that the chamber structure of the  hyperplane arrangement  I($\mathfrak{g},\mathbf{R}$) gives the geography of the extended K\"ahler-cone of a $G$-model with Lie algebra $\mathfrak{g}$ and associated representation $\mathbf{R}$. 
Indeed, we can leverage the description of the SU($2$)$\times$U($1$)-model  as a suspended pinch point  to fully control the structure of its crepant resolutions and compare it with the hyperplane arrangement I($\mathfrak{su}_2\oplus\mathfrak{u}_1,\mathbf{R}$).

One might think that the chamber structure of an SU($2$)$\times$U($1$)-model is given by the four chambers of the hyperplane arrangement 
  I($\mathfrak{u}_2, \bigwedge^2\oplus V$)=I($\mathfrak{u}_2, \bf{3}\oplus\bf{2})$ \cite{EJJN1}. However, that is not what is supported by the geometry.  Since the  suspended pinch point famously has three crepant resolutions connected by flops, we can anticipate that the same is true for the  SU($2$)$\times$U($1$)-model. We will show it explicitly via resolution of singularities and computing geometric weights of rational curves appearing over codimension-two points. 
The walls separating the chambers are given by the geometric weights of the singular fibers and the U($1$)-charge is computed also by intersection theory, using the Shioda-Tate-Wazir map \cite{Wazir}. 
All the hyperplane arrangements I($\mathfrak{su}_2\oplus\mathfrak{u}_1, \bf{3}_0\oplus \bf{2}_{m}$)  with $\bf{m}\neq 0$ have identical geographies: they have three chambers whose incidence graph is a Dynkin diagram of type A$_3$. The degenerate case $\bf{m}=0$ has a unique chamber.

\begin{figure}

\begin{tikzpicture}[scale=.6]
						\draw [fill=cyan!40!, draw=none, rotate=-135] (6,0) arc (0:180:6cm);;

\draw[ultra thick] (-4,0) to (4,0); 
\draw[ultra thick] (0,-4) to (0,4); 
\draw[ultra thick,rotate=-45] (0,0) to (0,4); 
\draw[ultra thick,rotate=-45] (0,0) to (4,0); 
\draw[ultra thick,rotate=-45] (0,0) to (0,-4); 
\draw[ultra thick,rotate=45] (0,0) to (0,4); 
\node[rotate=45] at  (41:4.5) {$x_1-x_2=0$};
\node at  (6:4.8) {$x_2=0$};
\node at  (-45:4.5) {$x_1+x_2=0$};
\node at  (-90:4.5) {$x_1=0$};

\node at (16,0){

\begin{tikzpicture}[scale=.6]
	\draw [fill=cyan!40!, draw=none, rotate=-90] (6.5,.05) arc (0:180:6.5cm);;
\draw[ultra thick, rotate=45] (0,0) to (0,5); 
\draw[ultra thick] (0,-5) to (0,5); 
\draw[ultra thick,rotate=-45] (0,0) to (0,5); 
\draw[ultra thick,rotate=-45] (0,0) to (5,0); 
\draw[ultra thick,rotate=-45] (0,0) to (0,-5); 
\node at  (51:5) {$\phi_1-\mu=0$};
\node at  (-48:5.1) {$-\phi_1-\mu=0$};
\node at  (-80:5.5) {$\phi_1=0$};
\end{tikzpicture}};

\end{tikzpicture}
\caption{Hyperplane arrangements 
I($\mathfrak{u}_2, \bigwedge^2 \oplus V$)  (on the left) and 
I($\mathfrak{su}_2\oplus \mathbf{u}_1, \bf{2}_{-1}$)  (on the right).
}

\end{figure}
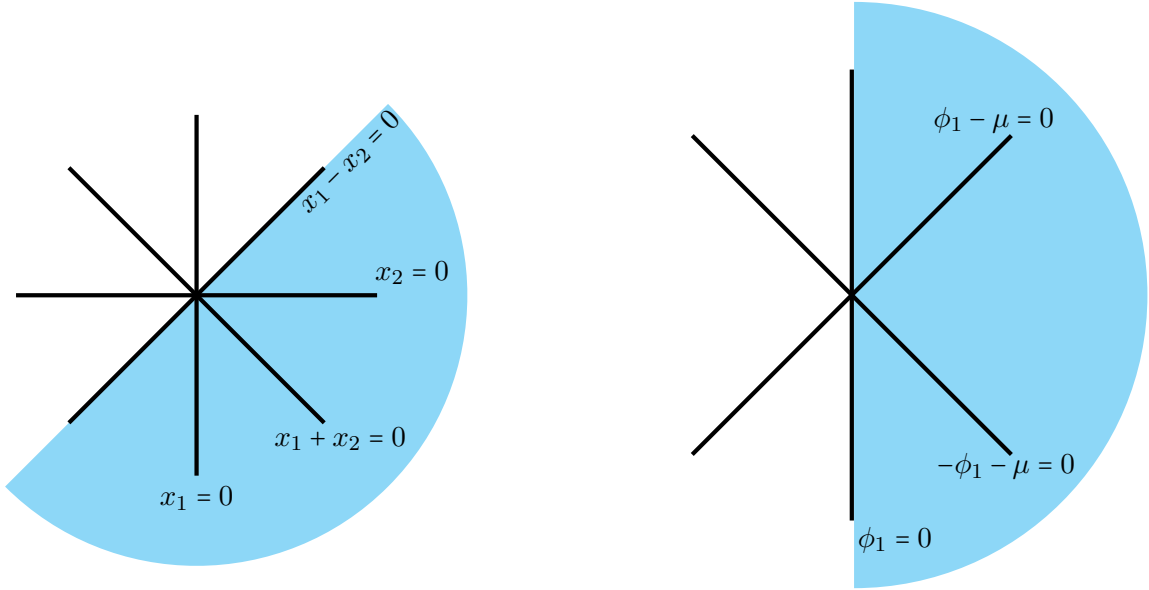

In Section \ref{Sec:Char}, we compute the Euler characteristic of an SU($2$)$\times$ U($1$)-model $Y$. 
In the case of a Calabi-Yau threefold, we also compute the Hodge numbers. 
These  data can be used to analyze anomaly cancellations in the six-dimensional theory obtained by a compactification of F-theory on a threefold  $Y$ or to compute the D3  tadpole on a fourfold $Y$.

\subsection{Weak coupling limit and tadpole matching condition}
The $j$-invariant has the following Fourrier $q$-expansion ($q=\exp(2\pi i \tau)$) with a simple pole at the cusp:  
\begin{equation}
j= \frac{1}{q}+744+196884 q +21493760 q^2 + \cdots
\end{equation}
where $\tau$ is the half-period of the elliptic curve. 
In type IIB string theory, $\tau$ is the axio-dilaton field \cite{Sen:1997gv}
\begin{equation}
\tau=C_0 + i e^{-\phi}, 
\end{equation}
 where $C_0$ is the type IIB axion and $\phi$ is the dilaton. The expectation value of $e^{\phi}$ is the string coupling $g_s$. 
When the string coupling becomes small, the $j$-invariant becomes infinite. This is called the {\em weak coupling limit} \cite{Sen:1997gv,CDE,AE1,AE2,Esole:2012tf,Clingher:2012rg}. 
It is interesting to consider the weak coupling limit of the  elliptic fibration introduced in equation \eqref{eq:SSPWM}. 
As explained in Section \ref{Sec:Weak}, we can define a weak coupling limit of the Weierstrass model of \eqref{eq:SSPWM} as follows: 
$$
s\to \epsilon s, \quad t\to h.
$$ 
In that limit, we define the double cover  $\rho: X\to B$ of the base $B$
\begin{equation}
X: \xi^2-h=0,
\end{equation}
where $\xi$ is a section of $\mathscr{L}$ and the equation $\xi^2-h=0$ is written in the total space of the line bundle $\mathscr{L}^{\otimes 2}$. 
In the weak coupling limit, we get the following brane spectrum: 
\begin{itemize}
\item an orientifold  wrapping the divisor $O: \xi=0$
\item two brane-image-brane wrapping  the divisor $O: \xi=0$
\item two invariant branes wrapping a smooth divisor $S=V(s)$ transverse to the orientifold. 
\end{itemize}
We show that this configuration satisfies the matching conditions between the D3 charge as computed in M-theory and in type IIB in absence of fluxes. 
\begin{equation}
2\chi(Y)= 4\chi(O)+2\chi(O)+2\chi(\overline{S}).
\end{equation}
We also show that this is a byproduct of a general identity true at the level of the Chow group: 
\begin{equation}
\varphi_* c(Y)= 3 \rho_*(O)+\rho_* (\overline{S}),
\end{equation}
where $\varphi: Y\to B$ is the elliptic fibration and $\rho: X\to B$ is the double cover of the base branched at $h=0$. 
That is particularly remarkable as only a handful of configurations are known to satisfy the tadpole matching condition without fluxes and its version at the level of Chow groups  \cite{CDE,AE1,AE2,EFY,EKY1}. 

\section{Preliminaries}

\subsection{$G$-models}
The F-theory algorithm provides a set of rules to attach to an elliptic fibration a triple $(G,\bf{R}, \mathfrak{g})$ where $G$ is a reductive Lie group, $\bf{R}$ is a representation of $G$, and $\mathfrak{g}$ is the Lie algebra of $G$. 
 The Lie algebra $\mathfrak{g}$  depends on the type of the dual graphs of the singular fibers over the generic points of the discriminant locus of the elliptic fibration. The weights of the representation $\mathbf{R}$ are  computed geometrically via intersection between rational curves that compose singular fibers over codimension-two points, and fibral divisors\footnote{Given an elliptic fibration $f:Y\to B$, a fibral divisor is a prime divisor $D$ such that there is a component $\Delta_i$ of the discriminant locus of $f$ such that $D$ is a prime component of $f^*\Delta_i$.}.

   An elliptic fibration associated via the F-theory dictionary to a gauge group $G$ is called a $G$-model. In a $G$-model with Lie algebra $\mathfrak{g}$ and a representation $\bf{R}$,   the pair $(\mathfrak{g},\bf{R})$ defines a hyperplane arrangement  I($\mathfrak{g},\bf{R})$ that controls several aspects of the birational geometry of the elliptic fibration such as the structure of its flops.
The incidence graph of the chambers of the hyperplane arrangement  I($\mathfrak{g},\bf{R})$ is conjectured to match the geography of Coulomb phases of a five-dimensional supergravity theory with eight supercharges (${\cal N}=1$) and hypermultiplets transforming in the representation $\mathbf{R}$ of the gauge group $G$. 

$G$-models are generally studied from a singular Weierstrass model satisfying the conditions of  Tate's algorithm so as to have the appropriate fibers over generic points of the discriminant locus and produce the dual graph of the Dynkin diagram of the Lie algebra of the  Langland dual of $G$. 
The smooth geometry is derived by a crepant resolution of singularities. 
Crepant resolutions are not unique when the elliptic fibration has dimension three or higher, however.  Distinct crepant resolutions are connected by a sequence of  flops. 
The structure defined by this network of flops can be predicted by studying the chamber structure of the hyperplane arrangement  I($\mathfrak{g},\bf{R})$ defined within the dual fundamental Weyl chamber of $\mathfrak{g}$ with interior walls given by the kernel of the weights of the representation $\mathbf{R}$.

Once the geometry is well-understood, it is possible to study in detail the cancellation of anomalies in an F-theory compactification on a $G$-model, and compute the prepotential in the Coulomb regime of the five-dimensional theory obtained by a compactification of M-theory on the same $G$-model. Several important  $G$-models are studied along these lines. 
The SU($2$), SU($3$), and SU($4$), models are treated in \cite{ESY1}.  The SU($5$)-model is discussed in \cite{EY,ESY2}.
The Spin($7$), Spin($8$), and G$_2$-models are analyzed in \cite{G2}, the  F$_4$-model in \cite{F4},  the E$_7$-model in \cite{E7}. 
The SO($3$), SO($5$), and SO($6$)-models have been studied recently in \cite{SO}.
The semi-simple cases of rank $2$ or $3$ with two simple components have also been analyzed recently:
the $\text{Spin($4$)}$ and SO($4$)-model is analyzed in \cite{SO4}, the $\text{SU($2$)$\times$G$_2$}$ in \cite{SU2G2}, the $\text{SU($2$)$\times$Sp($4$)}$ in \cite{Esole:2017hlw}, and the $\text{SU($2$)$\times$SU($3$)}$-model in \cite{SU2SU3}.

All of the models listed above confirm that the hyperplane arrangement I($\mathfrak{g},\bf{R})$ correctly predicts the geography of the flops of a $G$-model with Lie algebra $\mathfrak{g}$ and representation $\bf{R}$. However, they all have a Mordell--Weil group of rank zero. The case of a Mordell--Weil group of nonzero rank has yet to be explicitly considered.

\subsection{The group  SU($2$)$\times$U($1$)}

The group SU($2$)$\times$U($1$)  is the group of smallest rank among  compact, connected, reductive Lie groups with both a non-Abelian and an Abelian sector.  It is also one of the most celebrated in theoretical physics as it is the gauge group of the  electroweak sector of the Standard Model of particle physics.  In  the study of supersymmetric gauge theories in five-dimensional space-time with eight supercharges, the group SU($2$)$\times$U($1$) appears in the exceptional series of field theories with $N_f$ quark flavors and strongly coupled non-trivial fixed points studied by Morrison and Seiberg \cite{Morrison:1996xf}:
$$
E_8, \  E_7, \  E_6, \ E_5=\text{Spin($10$)}, \  E_4=\text{SU($5$)}, \  E_3=\text{ SU($2$)$\times$SU($3$)}, 
\ E_2=\text{ SU($2$)$\times$U($1$)}, \ E_1=\text{ SU($2$)}.
$$
We emphasize that the group SU($2$)$\times$U($1$) is not globally isomorphic to the group U($2$). 
There are strong geometric and physical arguments that U($2$) is the more natural group for the electroweak theory.
The group U($2$) is a $\mathbb{Z}_2$ quotient of SU($2$)$\times$U($1$), 
where $\mathbb{Z}_2$ is generated by minus the identity element of SU($2$)$\times$U($1$). 
Both groups share the same underlying Lie algebra $\mathfrak{su}_2\oplus\mathfrak{u}_1$:
$$
U(2)\cong \text{(SU($2$)$\times$U($1$))}/ \mathbb{Z}_2,\quad \mathfrak{u}_2\cong \mathfrak{su}_2\oplus\mathfrak{u}_1, \quad \mathbb{Z}_2\cong \{-\mathbb{I}, \mathbb{I}\}.
$$
Not all  representations of SU($2$)$\times$U($1$) project to   well-defined representations of U($2$). 
Continuous representations of U($2$) are the same as continuous representations of SU($2$)$\times$U($1$)  for which 
	$-\mathbb{I}$ is mapped to the identity.

Matter in representations $\bf{R}=\bf{3}_0\oplus \bf{2}_{\pm m}$ of   SU($2$)$\times$U($1$)  can  be obtained as follows.  
The branching rules for the decomposition of the adjoint representations of SU($3$) and USp($4$) along SU($2$)$\times$U($1$) are \cite{IMS}:
\begin{equation}
\bf{8}\to \bf{3}_0\oplus \bf{2}_{\pm{2}}\oplus \bf{1}_{0}, \quad \bf{10}\to \bf{3}_0\oplus \bf{2}_{\pm 1}\oplus \bf{1}_{\pm 2} \oplus\bf{1}_0.
\end{equation}

\section{The suspended pinch point as an  SU($2$)$\times$U($1$)-model}
In this section, we explain in some detail how the suspended pinch point gives an SU($2$)$\times$U($1$)-model.
We motivate our analysis by first considering a Whitney umbrella as a fibration of nodal curves. 
The SPP elliptic fibration is then derived by a suspension. Finally, we study the resulting geometry as a Weierstrass model whose Mordell--Weil group, discriminant, and $j$-invariant provide enough information to determine uniquely the associated group as an SU($2$)$\times$U($1$)-model.

 \subsection{The Whitney umbrella as a fibration of nodal cubics}

The Whitney umbrella is a rational surface, singular along a line of double points that worsen to a pinch point at the origin. Its defining equation is $k[x_0,x_1,x_2]/(x_0 x_1-x_2 x_3^2)$:
\begin{equation}
V(x_0^2 -x_1 x_2^2).
\end{equation}
It is a classical result that any smooth, complex, projective surface  is birational to a surface in $\mathbb{P}^3$ with only  ordinary singularities: a curve of double points, pinch points and triple points \cite[pages 616--618]{GriHa}. 
As a binomial variety with singularities in codimension-one, it is an example of a non-normal toric surface. 
The  Whitney umbrella appears naturally in F-theory:  Sen's weak coupling limit of F-theory is an orientifold theory with a seven-brane wrapping a complex surface with the singularities of a Whitney umbrella.

Since the equation of a Whitney umbrella is cubic, it is tempting to interpret  it as a fibration of curves of arithmetic genus one. 
It is natural to interpret a Whitney umbrella as an elliptic fibration: a Whitney umbrella is the surface swept by a nodal curve moving along a line. Then the singular fiber of this fibration is a cuspidal cubic at the origin. 
Under the linear transformation, $(x_0, x_1, x_2)\to (y, x+t,x)$, the Whitney umbrella becomes the following singular Weierstrass model
\begin{equation}
W_0: y^2-x^3-tx^2=0.
\end{equation}
We think of $W_0$ as a fibration over the line parametrized by $t$. 
However, $W_0$  is not an elliptic fibration in the formal sense since its generic fiber is not a smooth elliptic curve but a nodal curve. 
We can fix it by a deformation that will force the generic fiber to be a smooth elliptic curve. 
If we consider a deformation parameter  $s$ together with the parameter $t$ as defining a surface for the base of the fibration, we are naturally led to:
\begin{equation}
W'_0: y^2-x^3-tx^2-s=0.
\end{equation}
The variety W$'_0$ is now an elliptic threefold over the plane parametrized by $(s,t)$. 
We have increased the dimension of the base, but have lost the connection to the Whitney umbrella, however. 

Interestingly,  there is a satisfying solution to this problem if we do a base change that replaces the new deformation variable $s$ by its square $s^2$:  
\begin{equation}
Y_0:\quad y^2-x^3-tx^2-s^2=0.
\end{equation}
 Such a base change is similar to the one famously performed by Atiyah to obtain the threefold known as the conifold singularity in string theory and which gave the first example of a flop \cite{Atiyah}. What we end up with in our case is isomorphic to a singular threefold called the {\em  suspended pinch point} whose normal equation is the binomial variety:
 \begin{equation}
Z_0: \quad x_0 x_1 -x_2 x_3^2=0.
\end{equation}

\begin{figure}[htb]
\begin{center}
			\scalebox{.8}{\begin{tikzpicture}
			\draw[dashed, latex-latex, line width=1 pt] (1.2,0)--(3.2,0);  \draw[dashed, latex-latex, line width=1 pt]  (6.2,0)--(8.2,0);
						\draw [-latex, line width=1 pt] (.2,-1)--+(-45:1.8);
												\draw[->,->=stealth, ,line width=1 pt] (2.5,-4.1)--+(-45:1.7);
												\draw [-latex, line width=1 pt](7,-4.1)--+(225:1.7);
												\draw [-latex, line width=1 pt](9.5,-1.1)--+(225:1.7);
												\draw [-latex, line width=1 pt](4,-1.1)--+(225:1.7);\draw[-latex, line width=1 pt] (6,-1.1)--+(-45:1.7);
												\draw [-latex, line width=1 pt](12,-4.1)--+(200:5);
												\draw [-latex, line width=1 pt](-2,-4.1)--+(-20:5);
												\draw [-latex, line width=1 pt](-.5,-1.1)--+(225:1.7);
												\draw [-latex, line width=1 pt](10.5,-1.1)--+(-45:1.7);

					\node (1) at (0,0)	{\begin{tikzpicture}
				\node[draw,circle,thick,scale=1,fill=black,label=above:{}] (1) at (0,0){};
				\node[draw,circle,thick,scale=1,fill=black,label=above:{}] (2) at (1,0){};
				\node[draw,circle,thick,scale=1,fill=black,label=above:{}] (3) at (2,0){};
				\node[draw,circle,thick,scale=1,fill=black,label=above:{}] (5) at (0,1){};
				\node[draw,circle,thick,scale=1,fill=black,label=above:{}] (4) at (1,1){};
				\draw[thick] (1) to (2) to (3) to (4) to (5)  to (1);
				\draw[thick] (1) to (4) to (2);
			\end{tikzpicture}};
					\node (1) at (5,0)	{\begin{tikzpicture}
				\node[draw,circle,thick,scale=1,fill=black,label=above:{}] (1) at (0,0){};
				\node[draw,circle,thick,scale=1,fill=black,label=above:{}] (2) at (1,0){};
				\node[draw,circle,thick,scale=1,fill=black,label=above:{}] (3) at (2,0){};
				\node[draw,circle,thick,scale=1,fill=black,label=above:{}] (5) at (0,1){};
				\node[draw,circle,thick,scale=1,fill=black,label=above:{}] (4) at (1,1){};
				\draw[thick] (1) to (2) to (3) to (4) to (5)  to (1);
				\draw[thick] (2) to (4);   \draw[thick] (5) to (2);
			\end{tikzpicture}};
			\node (3) at (10,0) 
			{
			\begin{tikzpicture}
				\node[draw,circle,thick,scale=1,fill=black,label=above:{}] (1) at (0,0){};
				\node[draw,circle,thick,scale=1,fill=black,label=above:{}] (2) at (1,0){};
				\node[draw,circle,thick,scale=1,fill=black,label=above:{}] (3) at (2,0){};
				\node[draw,circle,thick,scale=1,fill=black,label=above:{}] (5) at (0,1){};
				\node[draw,circle,thick,scale=1,fill=black,label=above:{}] (4) at (1,1){};
				\draw[thick] (1) to (2) to (3) to (4) to (5)  to (1);
				\draw[thick] (5) to (3);
				\draw[thick] (5) to (2);
			\end{tikzpicture}};

			\node (4) at (2,-3) {
			\begin{tikzpicture}
				\node[draw,circle,thick,scale=1,fill=black,label=above:{}] (1) at (0,0){};
				\node[draw,circle,thick,scale=1,fill=black,label=above:{}] (2) at (1,0){};
				\node[draw,circle,thick,scale=1,fill=black,label=above:{}] (3) at (2,0){};
				\node[draw,circle,thick,scale=1,fill=black,label=above:{}] (5) at (0,1){};
				\node[draw,circle,thick,scale=1,fill=black,label=above:{}] (4) at (1,1){};
				\draw[thick] (1) to (2) to (3) to (4) to (5)  to (1);
				\draw[thick] (2) to (4);
			\end{tikzpicture}};
			
			\node (5) at (8,-3) {
						\begin{tikzpicture}
				\node[draw,circle,thick,scale=1,fill=black,label=above:{}] (1) at (0,0){};
				\node[draw,circle,thick,scale=1,fill=black,label=above:{}] (2) at (1,0){};
				\node[draw,circle,thick,scale=1,fill=black,label=above:{}] (3) at (2,0){};
				\node[draw,circle,thick,scale=1,fill=black,label=above:{}] (5) at (0,1){};
				\node[draw,circle,thick,scale=1,fill=black,label=above:{}] (4) at (1,1){};
				\draw[thick] (1) to (2) to (3) to (4) to (5)  to (1);
				\draw[thick] (2) to (5);
			\end{tikzpicture}};
			
			\node (6) at (5,-6){
			\begin{tikzpicture}
				\node[draw,circle,thick,scale=1,fill=black,label=above:{}] (1) at (0,0){};
				\node[draw,circle,thick,scale=1,fill=black,label=above:{}] (2) at (1,0){};
				\node[draw,circle,thick,scale=1,fill=black,label=above:{}] (3) at (2,0){};
				\node[draw,circle,thick,scale=1,fill=black,label=above:{}] (5) at (0,1){};
				\node[draw,circle,thick,scale=1,fill=black,label=above:{}] (4) at (1,1){};
				\draw[thick] (1) to (2) to (3) to (4) to (5)  to (1);				
			\end{tikzpicture}
			};

			\node (7) at (12,-3) {
						\begin{tikzpicture}
				\node[draw,circle,thick,scale=1,fill=black,label=above:{}] (1) at (0,0){};
				\node[draw,circle,thick,scale=1,fill=black,label=above:{}] (2) at (1,0){};
				\node[draw,circle,thick,scale=1,fill=black,label=above:{}] (3) at (2,0){};
				\node[draw,circle,thick,scale=1,fill=black,label=above:{}] (5) at (0,1){};
				\node[draw,circle,thick,scale=1,fill=black,label=above:{}] (4) at (1,1){};
				\draw[thick] (1) to (2) to (3) to (4) to (5)  to (1);
				\draw[thick] (3) to (5);
			\end{tikzpicture}};

			\node (8) at (-2,-3) {
						\begin{tikzpicture}
				\node[draw,circle,thick,scale=1,fill=black,label=above:{}] (1) at (0,0){};
				\node[draw,circle,thick,scale=1,fill=black,label=above:{}] (2) at (1,0){};
				\node[draw,circle,thick,scale=1,fill=black,label=above:{}] (3) at (2,0){};
				\node[draw,circle,thick,scale=1,fill=black,label=above:{}] (5) at (0,1){};
				\node[draw,circle,thick,scale=1,fill=black,label=above:{}] (4) at (1,1){};
				\draw[thick] (1) to (2) to (3) to (4) to (5)  to (1);
				\draw[thick] (1) to (4);
			\end{tikzpicture}};

\end{tikzpicture}
}
\end{center}
\caption{ {\bf Crepant resolutions and partial resolutions  of the suspended pinch point $Z_0: x_0 x_1-x_2 x_3^2=0$.} 
The suspended pinch point is at the bottom row. The three crepant resolutions are on the  top row.  The four partial resolutions are in the middle row. 
The external varieties of the middle row have the singularities of a cylindrical quadratic cone $k[x,y,z,t]/(z^2 - x y)$ while the others have singularities of a cone over a quadric surface (quadric threefold with a double point)  $k[x,y,z,t]/(xy-zt)$.
}
\end{figure}
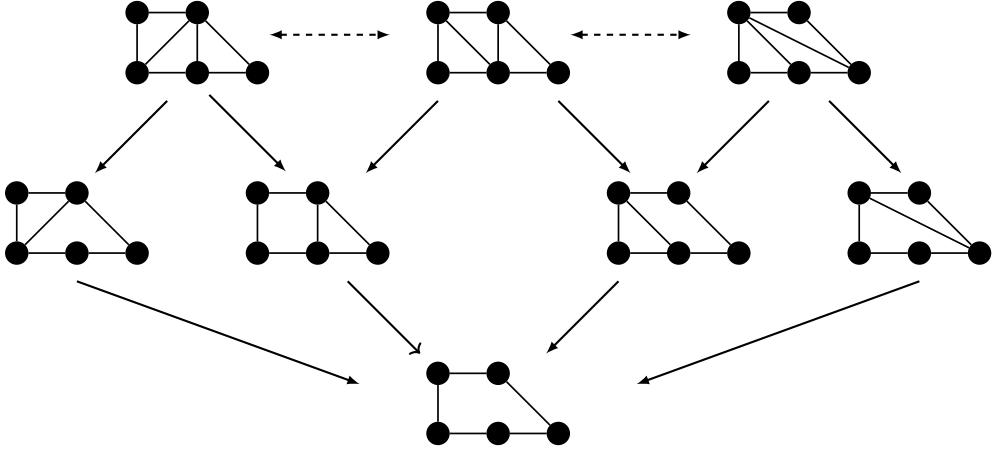

\subsection{Suspension}
\begin{defn}[Suspension]
Given a hypersurface singularity $X$ defined by an equation $f=0$ in the polynomial ring $k[x_1,\ldots, x_n]$, a suspension of $X$ is any hypersurface singularity with equation $f+y_1^2+\cdots +y_k^2$  in the extended polynomial ring $ k[x_1,\ldots, x_n,y_1,\ldots,y_k]$. 
\end{defn}  

The suspension of the Whitney umbrella by one variable is called the {\em suspended pinch point}.
Its defining equation after a trivial linear redefinition of the coordinates  is exactly $Z_0$.
Just as the suspended pinch point is the first suspension of the  Whitney umbrella, the conifold singularity ($x_0 x_1-x_2 x_3$) is the first suspension of an A$_1$ singularity $(xy-t^2$) after a trivial change of variables.
 
 \subsection{Weierstrass form, codimension-one  singular fibers, and Mordell--Weil group}
 
The suspended pinch point is a normal toric variety singular in codimension two along the ideal $x_0=x_1=x_3=0$.
This singularity has three distinct crepant resolutions connected by Atiyah flops.

By performing the shift
$x\rightarrow x-t/3$, the defining equation of the Weierstrass model $Y_0$ takes the form of a Weierstrass model 
\begin{equation}
y^2=x^3+fx+g,
\end{equation} 
with Weierstrass coefficients 
\begin{equation}
f=-\frac{t^2}{3},\quad g=s^2+\frac{2}{27}t^3,
\end{equation}
from which we can read off the discriminant  $\Delta=4 f^3+ 27 g^2$  and $j$-invariant  $j=1728 \cdot 4f^2/ \Delta$ \cite{Deligne.Formulaire}:
\begin{equation}
\Delta= s^2(4t^3+27 s^2), \quad \quad j=-\frac{2^8t^6}{s^2(4t^3+27 s^2)}.
\end{equation}
The reduced discriminant locus is composed of two prime divisors
\begin{equation}\label{Eq:DeltaJ}
S= V(s)~~\mathrm{and}~~~\Delta'=V(4t^3+27 s^2).
\end{equation}
While $S$ is smooth, $\Delta'$ has cuspidal singularities at $V(t,s)$, which is also the  intersection of the two divisors.
The valuation of $f,g,\Delta$ over $S$ is $(0,0,2)$, so  Tate's algorithm tells us that we have a singular fiber of Kodaira type I$_2$ over the generic point of $S$.  

Over a generic point, the valuation of $f,g,\Delta$ with respect to $\Delta'$ is $(0,0,1)$ so it is of Kodaira type I$_1$.

We note that the elliptic fibration has non-trivial rational sections 
\begin{equation}
\Sigma^\pm:\quad  x=y\pm s=0.
\end{equation}
These two sections correspond to opposite non-torsion elements of the Mordell--Weil group. Thus, the Mordell--Weil group is nontrivial with nonzero rank, that is generically one.

\section{Crepant resolutions and flops} \label{Sec:OneBU}

In this section, we study the crepant resolutions of the SPP elliptic fibration. 
As the SPP elliptic fibration can be written algebraically as a suspended pinch point, we anticipate that there are three distinct crepant resolutions forming a 3-chain worth of flops. 
The blowups used to describe the crepant resolutions of the suspended pinch points are usually centered at (non-Cartier) divisors as reviewed in Figure \ref{Fig:SSPRes}. 
However, for an elliptic fibration with a fiber of type I$_2$, we can also resolve the variety with a unique blowup 
centered on the singular locus $V(x,y,s)$ \cite{ESY1}. 
We will study these four resolutions and match them accordingly.

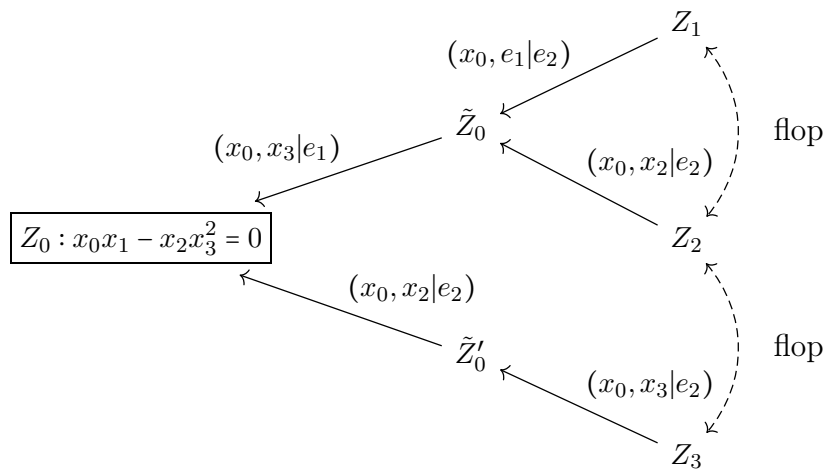
\begin{figure}[ht!]
\begin{center}
\begin{tikzcd}[column sep=60pt]
    &  & Z_1  
         \arrow[bend left=40, leftrightarrow, dashed]{dd}[right] {\quad \text{\large flop}}
    &\\ 
    & {
    \tilde{Z}_0}\arrow[leftarrow]{ru} {\displaystyle ({x}_0,e_1|e_2)}  \arrow[leftarrow]{rd} {\displaystyle ({x}_0,x_2|e_2)}& &\\
  \boxed{{Z}_0:x_0 x_1-x_2 x_3^2=0}\arrow[leftarrow]{ru} {\displaystyle (x_0,x_3|e_1)}   \arrow[leftarrow]{rd} {\displaystyle (x_0,x_2|e_2)}&  & Z_2  \arrow[bend left=40, leftrightarrow, dashed]{dd}[right] {\quad \text{\large flop}}&\\
     &
   {\tilde{Z}'_0}   \arrow[leftarrow]{rd} {\displaystyle ({x}_0,x_3|e_2)}&  \\
   & & Z_3
\end{tikzcd} 
\end{center}

\caption{{\bf The three crepant resolutions of the suspended pinch point $Z_0:x_0 x_1-x_2 x_3^2=0$}. 
In the blowup defining $Z_1$,  $e_1=0$ is the exceptional locus of the previous blowup defining $\tilde{Z}_0$, i.e. $e_1$ is the exceptional locus of the blowup of $Z_0$  centered at $x_0=x_3=0$. 
\label{Fig:SSPRes}}
\end{figure}

\subsection{A one-blowup crepant resolution}\label{Sec:ResZ2bis}
In this section, we resolve the singularities of the SPP elliptic fibration with one blowup centered on the support of the singular scheme. 
  We start from
\begin{equation}
F_0=-y^2+x^3+tx^2+s^2,
\end{equation}
defined over a base $B$, where $S=V(s)$ and $T=V(t)$ are two smooth Cartier divisors of $B$. 
The elliptic fibration is defined in the ambient space $X_0=\mathbb{P}_B[\mathscr{L}^{\otimes 3}\oplus \mathscr{L}^{\otimes 2}\oplus \mathscr{O}_B]$. 
The singular locus is the ideal $(x,y,s)$. 
We perform the blowup centered at the ideal $(x,y,s)$ \cite{ESY1}:
\begin{equation}\label{1}
\begin{tikzpicture}[baseline= (a).base]
\node[scale=.9] (a) at (0,0) {
\begin{tikzcd}[column sep=2.2cm, ampersand replacement=\&]
X_0 \arrow[leftarrow]{r} {\displaystyle \  (x,y,s|e_1)} \& X_1 
\end{tikzcd} 
};
\end{tikzpicture}
\end{equation}
giving $F_0=e_1^2 F_1$ with proper transform $F_1$: 
\begin{equation}
F_1=-y_1^2+e_1 x_1^3+tx_1^2+s_1^2.
\end{equation}
The blowup introduces a $\mathbb{P}^2$-projective bundle over the subvariety $V(x,y,s)$ of $X_0$ and the  projective coordinates  of the $\mathbb{P}^2$-bundle are $[x_1:y_1:s_1]$. Moreover, to follow the change in the fiber structure of the elliptic fibration, we note that after the blowup we have  $s=e_1 s_1$. 
We check by working in patches that this blowup provides a full resolution of all singularities. 
Moreover, since the blowup is a regular sequence of length three and the singularity is a double point, the canonical class does not change. Thus, by definition, the resolution is crepant. 

The fibral divisors are 
\begin{equation}
\begin{cases}
D_0:  s_1=-y_1^2+e_1 x_1^3+tx_1^2=0,\\
D_1:  e_1=-y_1^2+tx_1^2+s_1^2=0.
\end{cases}
\end{equation}
We denote by $C_0$ and $C_1$ the generic fibers of these fibral divisors. 
Over a generic point of $S$, $C_0$ and $C_1$ intersect as 
\begin{equation}
e_1=s=y_1^2-tx_1^2=0,
\end{equation}
which corresponds geometrically (i.e. after a field extension to allow us to consider the square root of $t$) to two  distinct points $y=\pm\sqrt{t} x_1$. 
We thus find the fiber is of type I$_2^{\text{ns}}$ with dual graph the affine Dynkin diagram $\tilde{A}_1$.

We expect the fiber to degenerate over $V(s,t)$, the intersection of $S$ and $\Delta'$.
When $t=0$, we see that $C_0$ is irreducible, while  $C_1$ splits into two lines:
\begin{equation}
(S\cap \Delta')\to 
\begin{cases}
C_0 &\rightarrow C_0,\\
C_1&\rightarrow C_1^++C_1^-,
\end{cases}
\end{equation}
where
\begin{equation}\begin{array}{ll}
C_1^\pm:  e_1=(s_1\pm y_1)=0.
\end{array}
\end{equation}
It follows that the  fiber over the generic point of $S\cap \Delta'$ is composed of three rational curves all meeting at the point $e_1=s_1=y_1=0$ and forming in this way a fiber of type IV$^{\text{s}}$. 
Thus, the collision of $S$ and $\Delta'$ defines the following enhancement 
\begin{equation}
\text{I}_1+\text{I}_2^{\text{ns}}\rightarrow \text{IV}^{\text{s}}.
\end{equation}
The intersection diagram of the fiber components form the nodes of the affine Dynkin diagram $\tilde{A}_2$. 
Figure \ref{fig:SPP.Fiber}  illustrates the  fiber structure of the SPP SU($2$)$\times$U($1$)-model.

		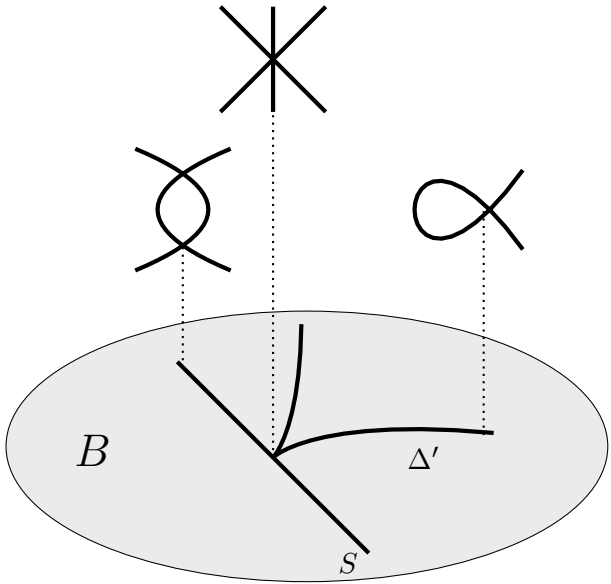
\begin{figure}[thb]
\begin{center}
			
				\begin{tikzpicture}[scale=1.1]
%
					 					
								\node at (0,4.8) {$ 
					  	\begin{tikzpicture}[scale=1]
						\draw[fill=black!8!] (.45,.85) ellipse (4cm and 1.8cm);
						\node at (-2.4,.8) {\LARGE $B$};
						\node at (1,-.7) {$S$};
						\node at (2,.7) {$\Delta'$};
							\draw[yshift=20,scale=.45,domain=-1.6:1.3,variable=\x, ultra thick, rotate=45] plot({2*\x*\x,\x*\x*\x});
							\draw[yshift=20,scale=.45,domain=-4:4,variable=\x, ultra thick,rotate=45] plot({0,-\x}); 
							\draw[dotted,thick] (0,.8) to (0,5.5){};
							
							 \node at (0,6) {
					  	\begin{tikzpicture}[scale=.7]
							\draw[ ultra thick] (-1,-1) to (1,1);
							\draw[ ultra thick] (-1,1) to (1,-1);
							\draw[ ultra thick] (0,-1) to (0,1);
						\end{tikzpicture}
								};

						\draw[dotted,thick] (2.8,1) to (2.8,4){};
							\draw[dotted,thick] (-1.2,3.5)--+ (-90:1.6){};

							\node at (2.6,4) {
							\begin{tikzpicture}[scale=2.5]
									\draw[ultra thick,scale=.4,domain=-1.2:1.2,variable=\x] plot({\x*\x-1,\x*\x*\x-\x-5});
							\end{tikzpicture}
								};
													
								 \node at (-1.2,4) {$ 
					  	\begin{tikzpicture}[scale=1.5]
							\draw[yshift=20,scale=.45,domain=-1.2:1.2,variable=\x, ultra thick] plot({\x*\x-.5,\x});
							\draw[yshift=20,scale=.45,domain=-1.2:1.2,variable=\x, ultra thick] plot({1-\x*\x-.5,\x}); 
						\end{tikzpicture}
								$};

						\end{tikzpicture}
								$};
								
				\end{tikzpicture}
			

			\end{center}
			\caption{
			Fiber structure of the elliptic fibration $Y_0: y^2z- x^3-t x^2z-s^2z^3=0$ after a crepant resolution. Here $S=V(s)$ and 
			$\Delta'=4 t^3 + 27 s^2$. 
			}
			\label{fig:SPP.Fiber}
		\end{figure}


The fibral divisor $D_0$ is a projective bundle $\mathbb P_{S}[ \mathscr{O}_S \oplus  \mathscr L]$ over $S$. 
The  fibral divisor $D_1$ is a conic bundle over $S$ with discriminant supported on $S\cap V(t)$. The  generic fiber of D$_1$ over $S$ is a  smooth conic that  splits ($S\cap\Delta'$)  into two lines meeting transversally at one point. 
The conic bundle can be described in this case as the vanishing locus of a section of $\mathscr{O}(2)\otimes \mathscr{L}^{\otimes 2}$ in 
$\mathbb P_{S}[ \mathscr{O}_S \oplus  \mathscr{ L}\oplus \mathscr{L}]$.

 The geometric weight of the curve C$_1^\pm$ 
 with respect to $(D_0,D_1)$ is $(C_1^\pm\cdot D_0, C_1^\pm \cdot D_1)$:
 \begin{equation}
 \varpi(C_1^\pm)=(1,-1),\quad \varpi(C_0)=(-2,2),
 \end{equation}
  which means that $C_1^\pm$ has weight $\varpi_1$ with respect to $\mathfrak{su}_2$ and we get  the  fundamental representation $\bf{2}$ for SU($2$). 

We will now determine the U($1$) charge of $C_1^\pm$. 
To prepare the stage, we first review how the sections change after the blowup. 
The center of the blowup is $(x,y,s)$, which is exactly the intersection of the  two sections $\Sigma^\pm$. Before the blowup, $\Sigma^\pm: x=y\pm s=0$. The total transform of $\Sigma^\pm$ is then  $e_1 x_1 = e_1(y_1\pm s_1)=0$, which consists of two components: 
the divisor $D_1$ given by $V(e_1)$ and the proper transform is
\begin{equation}
\Sigma_1^\pm:\quad x_1=y_1\pm s_1=0. 
\end{equation}
 The proper transform $\Sigma_1^\pm$ intersects $D_1$ transversally at  $e_1=x_1=y_1\pm s_1=0$, which is a point of $C_1^\pm$. The proper transforms  $\Sigma_1^+$ and $\Sigma_1^-$  do not intersect $C_0$ nor each other.

To compute the U($1$)-charge of $C_1^\pm$, we determine the divisor $\Gamma$ that is the image of $\Sigma^\pm_1$ under the Shioda map, which is essentially an orthogonal projection.  In the present case, one can check that the following divisor does the job:
\begin{equation}
\Gamma=\Sigma_1^+ -\Sigma_0-\pi^*\pi_* [(\Sigma^+_1-\Sigma_0)\cdot  \Sigma_0]+\frac{1}{2}D_1=\Sigma_1^+ -\Sigma_0-\pi^* L+\frac{1}{2}D_1.
\end{equation}
The divisor $\Gamma$  satisfies all the conditions required for the Shioda-Tate-Wazir map \cite{Morrison:2012ei}: vanishing intersection with the generic fiber, vanishing intersection with any fibral divisor,  vanishing intersection with any class pulled back from the base.  
The sections $\Sigma_1^+$ and $\Sigma_0$ do not intersect, the components $C_0$ and $C_1$ project to points in the base. 
We can now compute the U($1$) charge of $C_1^\pm$ as an intersection number 
\begin{equation}
\Gamma \cdot C_1^{\pm}=\pm 1/2.
\end{equation}
We will write the U($1$) charges in terms of units of this charge, we use  $2\Gamma$ to ensure that the  U($1$)-charge is  integral. 
Thus, we get:
\begin{equation}
(D_0, D_1, 2\Gamma) \cdot C_1^\pm= (1,-1, \pm1).
\end{equation}

\subsection{Crepant resolution $Z_1$}

Starting from the expression for the suspended pinch point in~\eqref{eq:xiSPP}:
\begin{equation}
Z_0:x_0 x_1-x_2 x_3^2,
\end{equation}
we would now like to resolve the singularity of the elliptic fibration.  We will follow the blowup sequences depicted in Figure~\ref{Fig:SSPRes}, beginning with the one that leads to $Z_1$ here.

We first perform the blowup centered at the ideal $(x_0,x_3)$:
  \begin{equation}\label{12}
\begin{tikzcd}[column sep=2cm, ampersand replacement=\&]
Z_0 \arrow[leftarrow]{r} {\displaystyle (x_0,x_3|e_1)} \& \tilde{Z}_0 
\end{tikzcd} 
\end{equation}
giving  $Z_0=e_1\tilde{Z}_0$ with
   \begin{equation}\label{eq:z0p}
\tilde{Z}_0=x_0 x_1-e_1 x_2 x_3^2,
\end{equation}
 Here and in what follows we will exert a slight abuse of notation by 
using the same variable name for the $x_i$ both before and after any blowup.
The blowup introduces a $\mathbb{P}^1$-projective bundle over the subvariety $V(x_0,x_3)$ and the  projective coordinates  of the $\mathbb{P}^1$-bundle are $[x_0:x_3]$. 

We then perform a second blowup centered at the ideal $(x_0,e_1)$:
   \begin{equation}\label{11}
\begin{tikzcd}[column sep=2cm, ampersand replacement=\&]
Z_0 \arrow[leftarrow]{r} {\displaystyle (x_0,x_3|e_1)} \& \tilde{Z}_0 \arrow[leftarrow]{r} {\displaystyle (x_0,e_1|e_2)} \& Z_1 .
\end{tikzcd} 
\end{equation}
These two blowups add up to the following birational transformation
\begin{equation}
(x_0,x_1,x_2,x_3)\rightarrow (e_1 e_2 x_0, x_1, x_2, e_1 x_3),
\end{equation}
and the projective coordinates of the full resolution are $[e_2 x_0:x_3][x_0:e_1]$.  
The total transform of $Z_0$ is  $Z_0=e_1e_2^2 Z_1$, where $Z_1$ is its proper transform given by 
\begin{equation}
Z_1=x_0 x_1-e_1 x_2 x_3^2.
\end{equation}
This blowup introduces another $\mathbb{P}^1$-projective bundle over the subvariety $V(x_0,e_1)$.  
One can check that that this blowup provides a full resolution of all singularities. The blowup is trivially crepant as we only blowup smooth divisors. 

We will now understand the geometry of this crepant resolution from the point of view of the elliptic fibration. 
We first recall the  transformed coordinates before the two  blowups:
\begin{equation}\label{eq:transformcoor}
\begin{cases}
x_0= y-s,\  x_1= y+s\  x_2= x+t,\   x_3=x,\\
 x=x_3, \quad
y=\frac{1}{2}(x_0+x_1), \quad s=\frac{1}{2}(x_1-x_0), \quad t =x_2-x_3.
\end{cases}
\end{equation}
We note that the singular locus $V(y,s,t)$ is $V(x_0,x_1,x_3)$, the section $\Sigma^-=V(y-s,x)$ is $V(x_1, x_3)$, the section $\Sigma^+=V(y+s, x)$ is $V(x_0,x_3)$, and the divisor $S=V(s)$ is $V(x_1-x_0)$. 
We also recall that the singular locus is the intersection $\Sigma^+\cap\Sigma^-$:
\begin{equation}
\begin{array}{l}
\Sigma^+=V(y+s,x)=V(x_0,x_3), \quad \Sigma^-=V(y-s,x)=V(x_1,x_3), \\
 S=V(s)=V(x_1-x_0), \quad V(s,t)=V(x_1-x_0, x_2-x_3).
\end{array}\end{equation}
It follows that the center of the first blowup is the section $\Sigma^+$ while the center of the second blowup is the intersection of the exceptional locus of the first blowup  with  the section $\Sigma^-$.

After the two blowups, we have
\begin{equation}
\begin{cases}
x=e_1 e_2 x_3\\
y=\frac{1}{2}(e_1 e_2^2 x_0+x_1)\\
s=\frac{1}{2}( x_1-e_1 e_2^2x_0)\\
t=x_2-e_1 e_2 x_3.
\end{cases}
\end{equation}
So that when $s=0$, $x_1=e_1 e_2^2 x_0$ and $Z_1$ factors into two components
\begin{equation}
Z_1=e_1(e_2^2 x_0^2-x_2 x_3^2).
\end{equation}
Thus, the fibral divisors over $S$ are
\begin{equation}
\begin{cases}
D_0: x_1-e_1 e_2^2x_0=e_2^2 x_0^2-x_2x_3^2=0\\
D_1:  x_1=e_1=0.\\
\end{cases}
\end{equation}
We denote by $C_0$ and $C_1$ the generic fiber of the fibral divisor $D_0$ and $D_1$, respectively. 
Over a generic point of $S$, $C_0$ and $C_1$ intersect as 
\begin{equation}
x_1=e_1=e_2^2 x_0^2-x_2x_3^2=0
\end{equation}
which corresponds geometrically (i.e. after a field extension to allow us to consider the square root of $t$) to two  distinct points $e_2 x_0=\pm\sqrt{x_2} x_3$. 
We thus find the fiber is of type I$_2^{\text{ns}}$ with dual graph the affine Dynkin diagram $\widetilde{\text{A}}_1$.

We expect the fiber to degenerate  further over $V(s,t)$, the intersection of $S$ and $\Delta'$.
When $t=0$, we have
$x_2=e_1 e_2 x_3$ and the second equation defining $C_0$ factors as $e_2(x_0^2-e_1 x_3^3)=0$. Thus, 
we see that $C_1$ is irreducible, while  $C_0$ splits into two lines:
\begin{equation}
(S\cap \Delta')\to 
\begin{cases}
C_0 &\rightarrow C_0'+C_2,\\
C_1&\rightarrow C_1,
\end{cases}
\end{equation}
where
\begin{equation}
C_0':  x_1=x_2-e_1 e_2 x_3=e_2 x_0^2-e_1 x_3^3=0, \quad
C_{2}:  x_1=x_2  =e_2=0.
\end{equation}
It follows that the  fiber over the generic point of $S\cap \Delta'$ is composed of three rational curves all meeting at the point $V(s, t,e_1,e_2)=V(x_1,x_2,e_1,e_2)$ and forming in this way a fiber of type IV$^{\text{s}}$. 
Thus, the collision of $S$ and $\Delta'$ defines the following enhancement 
\begin{equation}
\text{I}_1+\text{I}_2^{\text{ns}}\rightarrow \text{IV}^{\text{s}}.
\end{equation}
The geometric weight of the curve $C_1$ is  $(2,-2)$ while the geometric weight of $C_2$ is $(-1,1)$. 
\begin{equation}
\varpi(C'_0)=\varpi(C_2)=(-1,1), \quad 
\varpi(C_1)=(2,-2)
.
\end{equation}
We notice that $\varpi(C_2)$ does not touch the zero section and has  exactly the opposite weight we got for the resolution of section \ref{Sec:ResZ2bis}. We can conclude that the two are connected by a flop.

\subsection{Crepant resolution $Z_2$}
To perform the second blowup in Figure~\ref{Fig:SSPRes}, we follow the above steps up to equation~(\ref{eq:z0p}) after which we perform an alternate blowup, centered at the ideal $(x_0,x_2)$:
 \begin{equation}\label{112}
\begin{tikzpicture}[baseline= (a).base]
\node[scale=.9] (a) at (0,0) {
\begin{tikzcd}[column sep=2cm, ampersand replacement=\&]
Z_0 \arrow[leftarrow]{r} {\displaystyle (x_0,x_3|e_1)} \& \tilde{Z}_0 \arrow[leftarrow]{r} {\displaystyle (x_0,x_2|e_2)} \& Z_2 
\end{tikzcd} 
};
\end{tikzpicture}
\end{equation}
 giving   $Z_0=e_1 e_2 Z_2$ with
\begin{equation}
Z_2=x_0x_1-e_1 x_2 x_3^2.
\end{equation}
This blowup introduces another $\mathbb{P}^1$-projective bundle over the subvariety $V(x_0,x_2)$ of $X$.  The projective coordinates of the full resolution are $[e_2 x_0:x_3][x_0:x_2]$.  
One can check that this blowup provides a full resolution of all singularities and is crepant.

In terms of the transformed coordinates
\begin{equation}\label{eq:transformcoor}
x_0\rightarrow y-s,~~x_1\rightarrow y+s,~~x_2\rightarrow x+t,~~x_3\rightarrow x,
\end{equation}
 we find that after these blowups, under which
\begin{equation}
(x_0,x_1,x_2,x_3)\rightarrow (e_1 e_2 x_0, x_1, e_2 x_2, e_1 x_3),
\end{equation}
we have
\begin{equation}\left\{\begin{array}{l}\label{eq:co1}
y=\frac{1}{2}(e_1 e_2 x_0+x_1)\\s=\frac{1}{2}( x_1-e_1 e_2 x_0)\\x=e_1 x_3\\t=e_2 x_2-e_1 x_3.
\end{array}\right.\end{equation}
So that when $s=0$
\begin{equation}
Z_2=e_1 (e_2 x_0^2-x_2 x_3^2).
\end{equation}
The fibral divisors are 
\begin{equation}
\begin{cases}
D_0:  x_1-e_1 e_2 x_0=e_2 x_0^2-x_2 x_3^2=0\\
D_1:  x_1=e_1=0 .\\
\end{cases}
\end{equation}
We denote by $C_0$ and $C_1$ the generic fibers of these fibral divisors. 
Over a generic point of $S$, $C_0$ and $C_1$ intersect as 
\begin{equation}
x_1=e_1=e_2 x_0^2-x_2 x_3^2=0\\
\end{equation}
which corresponds geometrically  to two  distinct points. 
We thus find the fiber is of type I$_2^{\text{ns}}$ with dual graph the affine Dynkin diagram $\tilde{A}_1$.

We expect the fiber to degenerate over $V(s,t)$, the intersection of $S$ and $\Delta'$,
\begin{equation}
C_1 (\text{at $t=0$}):  x_1=e_1=e_2 x_2=0
\end{equation}
and the projective coordinates $[x_0:x_2][e_2 x_0:x_3]$. The first $\mathbb{P}^1$ only appears over the exceptional $e_2=0$ which would imply $t=0$.  So, away from $t=0$ we have the $\mathbb{P}^1$ parameterized by $[e_2 x_0:x_3]$.  When $t=e_2 x_2=0$ we have a union of two lines $[x_0:x_2][0:1]$ and $[1:0][e_2:x_3]$.   Meanwhile
\begin{equation}
C_0  (\text{at $t=0$}):   x_1-e_1 e_2 x_0=e_2 x_0^2-x_2 x_3^2=e_2 x_2-e_1 x_3=0.
\end{equation}
When $t=0$, we thus see that $C_0$ is irreducible, while  $C_1$ splits into two lines:
\begin{equation}
(S\cap \Delta')\to 
\begin{cases}
C_0 &\rightarrow C_0,\\
C_1&\rightarrow C_1'+C_{12},
\end{cases}
\end{equation}
where
\begin{equation}
C_{12}:  s=e_1=e_2=0,\quad 
C_1':  s=e_1=x_2=0.
\end{equation}
All three curves ($C_0$, $C_1'$, $C_{12}$) meet at a single point $s=e_1=e_2=x_2=0$. Thus, they form a fiber of type IV$^{\text{s}}$ and  
 the collision of $S$ and $\Delta'$ defines the following enhancement 
\begin{equation}
\text{I}_1+\text{I}_2^{\text{ns}}\rightarrow \text{IV}^{\text{s}}.
\end{equation}
The geometric weights of the curves  $C'_1$ and $C_{12}$ with respect to the fibral divisors $D_0$ and $D_1$ are 
\begin{equation}
\varpi(C'_1)=\varpi(C_{12})=(1,-1), \quad \varpi(C_0)=(-2,2).
\end{equation}

\subsection{Crepant resolution $Z_3$}
Performing the third sequence of blowups in Figure~\ref{Fig:SSPRes}, we begin with a blowup centered at the ideal $(x_0,x_2)$:
 \begin{equation}\label{113}
\begin{tikzpicture}[baseline= (a).base]
\node[scale=.9] (a) at (0,0) {
\begin{tikzcd}[column sep=1.4cm, ampersand replacement=\&]
Z_0 \arrow[leftarrow]{r} {\displaystyle (x_0,x_2|e_1)} \& \tilde{Z}'_0 
\end{tikzcd} 
};
\end{tikzpicture}
\end{equation}
giving $Z_0=e_1 \tilde{Z}_0'$ with
\begin{equation}
\tilde{Z}_0'=x_0 x_1-x_2 x_3^2.
\end{equation}
The blowup introduces a $\mathbb{P}^1$-projective bundle over the subvariety $V(x_0,x_2)$ and the  projective coordinates  of the $\mathbb{P}^1$-bundle are $[x_0:x_2]$.  

We then perform a second blowup centered at the ideal $(x_0,x_3)$:
 \begin{equation}\label{114}
\begin{tikzpicture}[baseline= (a).base]
\node[scale=.9] (a) at (0,0) {
\begin{tikzcd}[column sep=1.8cm, ampersand replacement=\&]
Z_0 \arrow[leftarrow]{r} {\displaystyle (x_0,x_2|e_1)} \& \tilde{Z}'_0 \arrow[leftarrow]{r} {\displaystyle (x_0,x_3|e_2)} \& Z_2 
\end{tikzcd} 
};
\end{tikzpicture}
\end{equation}
giving $Z_0=e_1 e_2 Z_3$ with
\begin{equation}
Z_3= x_0 x_1-e_2 x_2 x_3^2.
\end{equation}
This blowup introduces another $\mathbb{P}^1$-projective bundle over the subvariety $V(x_0,x_3)$ of $X$ .  The projective coordinates of the full resolution are  $[e_2 x_0:x_2][x_0:x_3]$.  
One can check that that this blowup provides a full resolution of all singularities and is crepant.

In terms of the transformed coordinates
\begin{equation}\label{eq:transformcoor}
x_0\rightarrow y-s,~~x_1\rightarrow y+s,~~x_2\rightarrow x+t,~~x_3\rightarrow x,
\end{equation}
 we find that after these blowups, under which
\begin{equation}
(x_0,x_1,x_2,x_3)\rightarrow (e_1 e_2 x_0, x_1, e_1 x_2, e_2 x_3),
\end{equation}
we have
\begin{equation}\left\{\begin{array}{l}\label{eq:co3}
x= e_2 x_3\\
y=\frac{1}{2}(e_1 e_2 x_0+x_1)\\
s=\frac{1}{2}( x_1-e_1 e_2 x_0)\\
t=e_1 x_2-e_2 x_3.
\end{array}\right.\end{equation}
When $s=0$,  $Z_3$ factors into two components: 
\begin{equation}
Z_3=e_2(e_1 x_0^2-x_2 x_3^2).
\end{equation}
Thus, the fibral divisors are 
\begin{equation}
\begin{cases}
D_0:  x_1-e_1 e_2 x_0=e_1 x_0^2-x_2 x_3^2=0\\
D_1:  x_1=e_2=0.\\
\end{cases}
\end{equation}
We denote by $C_0$ and $C_1$ the generic fibers of these fibral divisors. 
Over the generic point of $S$,  $C_0$ and $C_1$ are rational curves. But as we shall see $C_0$ degenerates after a specialization.

Over a generic point of $S$, $C_0$ and $C_1$ intersect as 
\begin{equation}
x_1=e_2=e_1 x_0^2-x_2 x_3^2=0
\end{equation}
which corresponds geometrically  to two  distinct points. We thus find the fiber is of type I$_2^{\text{ns}}$ with dual graph the affine Dynkin diagram $\tilde{A}_1$.

We expect the fiber to degenerate over $V(s,t)$, the intersection of $S$ and $\Delta'$.  When $t=0$, we see that $C_1:~ x_1=e_1=e_2=0$ is irreducible, while  $C_0$  is given by 
\begin{equation}\label{eq:Z3C0T}
C_0: ~x_1-e_1 e_2 x_0 =e_1 x_2 -e_2 x_3=e_1 x_0^2-x_2 x_3^2=0.
\end{equation}
We now consider a particular linear combination of the second and third equation:
\begin{equation}
-x_0^2(e_1 x_2 -e_2 x_3)+x_2(e_1 x_0^2-x_2 x_3^2)=
x_3(e_2 x_0^2-x_2^2 x_3)=0.
\end{equation} 
This shows that over $S\cap \Delta'$, the curve  $C_0$ becomes  reducible. 
We note that $x_3=0$ implies\footnote{When $x_3=0$, the third part of equation 
\eqref{eq:Z3C0T} implies that $e_1x_0^2=0$.  Since $(x_0,x_3)$ cannot vanish at the same time, $x_3 e_1 x_0^2=0$ 
implies that  $x_3=e_1=0$. Finally,  $x_3=e_1=x_1-e_1 e_2 x_0=0$  forces $x_3=e_1=x_1=0$.
} that $x_3=e_1=x_1=0$,  which is a rational curve.  
Hence, we get 
\begin{equation}
(S\cap \Delta')\to 
\begin{cases}
C_0 &\rightarrow C'_0+C_2,\\
C_1&\rightarrow C_1,
\end{cases}
\end{equation}
with 
 \begin{equation}\begin{cases}
C'_{0}:  ~x_1-e_1 e_2 x_0 = e_1 x_2-e_2 x_3 = x_2^2 x_3 -x_0^2 e_2=e_1 x_0^2-x_2 x_3^2=0\\
C_{2}:  ~x_3=x_1=e_1=0\\
C_{1}:  ~ e_2=x_1=e_1=0.
\end{cases}
\end{equation}
The curve  $C'_0$ is the normalization of a cuspidal curve. In the patch $x_3\neq 0$:
$$
x_1=\frac{e_1^3 x_0^3}{x_3^3},\quad e_2= \frac{e_1^2 x_0^2}{x_3^3}, \quad x_2=e_1 \frac{ x_0^2}{x_3^2}.
$$ 
All three curves are therefore smooth rational curves and meet at the same point $x_1=x_3=e_1=e_2=0$. 
Thus, they form a  fiber of type IV$^{\text{s}}$ and the collision of $S$ and $\Delta'$ defines the following enhancement 
\begin{equation}
\text{I}_1+\text{I}_2^{\text{ns}}\rightarrow \text{IV}^{\text{s}}_{2}.
\end{equation}
The geometric weights of the curves  $C'_0$ and $C_2$ with respect to the fibral divisors $D_0$ and $D_1$ are 
\begin{equation}
\varpi(C'_0)=\varpi(C_2)=(-1,1), \quad \varpi(C_1)=(2,-2).
\end{equation}
We note that the weights obtained here are  of opposite signs compare to those of $Z_2$ but  coincide with those of  $Z_1$.

\section{Euler characteristic and Hodge numbers}\label{Sec:Char}

In this section, we compute the Euler characteristic of an SU($2$)$\times$U($1$)-model following \cite{Euler}.
We also compute the Hodge numbers in the case of a  Calabi--Yau  threefold.

\begin{thm}[{Euler characteristic of an SU($2$)$\times$U($1$)-model}]\label{Thm:Euler}
Let $\mathscr{L}$ be a line bundle over a projective variety $B$. 
Consider the projective bundle $X_0=\mathbb{P}_B[\mathscr{O}_B\oplus\mathscr{L}^{\otimes 2}\oplus\mathscr{L}^{\otimes 3}]$. 
Let $Y_0$ be the elliptic fibration of an {\textup{SU($2$)$\times$U($1$)}}-model cut by the equation 
$y^2z-(x^3+ t x^2 z +  s^2 z^3)=0$ in $X_0$. Let $Y$ be a crepant resolution of $Y_0$.  Then the Euler characteristic of $Y$ is 
\begin{align}
\chi(Y)& =\frac{12 L}{1+3L} c(TB),\\
 &= 
 12 L+  12 (c_1 L - 3L^2) + 12 L (c_2 - 3 c_1 L + 9 L^2)+\cdots 
\end{align}
where $c(TB)$ is the total Chern class of the base $B$ and $L=c_1(\mathscr{L})$, $c_1$ is written for $c_1(TB)$, and the Euler characteristic is given by the terms of degree $\dim B$. 
 \end{thm}
\begin{proof}
Since the Euler characteristic is a k-invariant, it is preserved between two smooth varieties connected by a crepant birational map \cite{Batyrev.Betti}. 
 Thus, to compute the 
 Euler characteristic of an SU($2$)$\times$U($1$)-model, it is enough to compute it for the crepant resolution of our choice. 
 We will use the crepant resolution of Section \ref{Sec:OneBU}. 
The center of the blowup is exactly the same as for an SU($2$)-model with the specialization $[S]=3L$. 
The Euler characteristic of an SU($2$)-model is \cite[\S 5]{Euler}:
$$
\chi(Y)=6 \frac{3 L S+2 L-S^2}{(S+1) (1+6 L-2 S)} c(TB).
$$
We get the Euler characteristic of an SU($2$)$\times$U($1$)-model  by substituting $S\to 3L$ in the previous formula. 
\end{proof}
\begin{rem}
We note that this Euler characteristic matches that of the crepant resolution of a generic Weierstrass-model with Mordell--Weil group $\mathbb{Z}/3\mathbb{Z}$ \cite[\S 7.3]{Salazar} corresponding to the PSU($3$)-model \cite[\S2.5]{Char2} and of an elliptic fibration of type E$_6$ (see \cite[Theorem 4.3]{AE2} or \cite{Char2}). 
\end{rem}

In the Calabi--Yau case, we have $c_1=L$, which gives the generating function 
\begin{equation}
\chi(Y)= 12 c_1- 24 c_1^2+ 12(6 c_1^3 + c_1 c_2) +\cdots
\end{equation}

Using motivitic integration, Kontsevich proved in \cite{Kontsevich} that  birational equivalent Calabi--Yau varieties have the same 
Hodge numbers.  We now assume that the variety is a Calabi--Yau threefold. 
We use the following definition of a Calabi--Yau variety: 
\begin{defn}
A Calabi-Yau variety is a smooth compact projective variety  $Y$ of dimension $n$
with a trivial canonical class and such that $H^i(Y, \mathscr{O}_Y)=0$  for $1 \leq  i \leq n-1$.
\end{defn}
In the case of a threefold, such a Calabi--Yau has non-trivial Hodge numbers $h^{1,1}(Y)$ and $h^{2,1}(Y)$, and its Euler characteristic satisfies the relation 
  $\chi(Y)=2 h^{1,1}(Y)-2 h^{2,1}(Y)$.  
  Following \cite[\S4]{Euler}, we have the following Corrollary. 
\begin{cor}[Hodge numbers for a Calabi--Yau threefold SPP-model]  
Under the assumptions  of Theorem \ref{Thm:Euler}, if $Y$ is a Calabi--Yau threefold, 
its Euler characteristic is $\chi(Y)=-24 K^2$, where $K$ is the canonical class of the base of the elliptic fibration.  
Moreover, its Hodge numbers $h^{1,1}(Y)$ and $h^{2,1}(Y)$ are 
$$
\chi(Y)=-24K^2, \quad h^{1,1}(Y)=13-K^2, \quad h^{2,1}(Y)=13+11 K^2.
$$
\end{cor}
\begin{proof}
We compute $h^{1,1}(Y)$  using the Shioda--Tate--Wazir theorem \cite[Corollary 3.2]{Wazir}: 
$$
h^{1,1}(Y)=1+h^{1,1}(B)+r+f,
$$
where $r$ is the Mordell--Weil rank and $f$ is the number of geometrically irreducible fibral divisors not touching the zero section of the elliptic fibration. 
The Calabi--Yau condition forces $B$ to be a rational surface. Denoting the canonical class of $B$ by $K$, Noether's formula gives 
$$
h^{1,1}(B)= 10-K^2.
$$
In the present case, $r=f=1$. We can compute h$^{1,1}(Y)$ and then use the Euler characteristic $\chi(Y)=2 h^{1,1}(Y)-2 h^{2,1}(Y)$ to compute $h^{2,1}(Y)$. 
For a Calabi-Yau, the vanishing of the first Chern class implies that $L=-K$. 
\end{proof}

\section{Weak coupling limit}\label{Sec:Weak}

\subsection{Brane geometry at weak coupling}
When taking a weak coupling limit, the discriminant locus can split into different components that are wrapped by orientifolds and  branes. These branes can be singular and can split further into brane-image-brane pairs in the double cover of the base. 
The typical situation is the following. Upon a weak coupling limit, the discriminant and the  $j$-invariant are at leading order in the deformation parameter $\epsilon$: 
\begin{align}
\Delta &= \epsilon^2 h^{2+n} \prod_i  (\eta^2_i-h \psi_i^2)\prod_j (\eta^2_j -h \chi_j), \prod_k \phi_k+O(\epsilon^3),\\
 j &\propto \frac{h^{4-n}}{\epsilon^2  \prod_i  (\eta^2_i-h \psi_i^2)\prod_j (\eta^2_j -h \chi_j) \prod_k \phi_k}.
\end{align}
The locus $h=0$ is the orientifold locus as seen from the base of the elliptic fibration and is a section of  the line bundle $\mathscr{L}^{\otimes 2}$. 
As $\epsilon$ goes to zero, $j$ goes to infinity and the string couplings goes to zero:
\begin{align}
\lim_{\epsilon\rightarrow 0} j=\infty\Longrightarrow Im(\tau)=\infty\iff g_s=0.
\end{align}
The  orientifold theory is defined by considering the double cover  $\rho: X\to B$ of the base  $B$ branched at $h=0$. Explicitly, we have 
\begin{equation}
X:   \quad \xi^2=h.
\end{equation}
The involution map $\sigma:X\rightarrow X$, which sends $\xi$ to $-\xi$, can be used to define a $\mathbb{Z}/2\mathbb{Z}$ orientifold symmetry $ \Omega (-)^{F_L}\sigma$ and the branched locus $\xi=0$ is  interpreted as an $O7$ orientifold. 
The geometry of Sen's weak coupling limit can be summarized by the following table:

\begin{table}[htb]
\begin{center}
\begin{tabular}{|c|l|l|}
\hline
Name & In the discriminant & In the double cover $X$\\
\hline
Orientifold & $h^2$ & $\xi=0$ \\
\hline
 Whitney brane &  $\eta^2_i-h \chi_i $ & $\eta^2_i-\xi^2 \chi_i=0$ \\
\hline
 Brane-image-brane pair & $\eta^2_j-h \psi^2_j$ & $(\eta_j+ \xi \psi_j)(\eta_j- \xi \psi_j)=0$\\
\hline
 Invariant  brane &  $\phi_k$ & $\phi_k=0$\\
\hline
\end{tabular}
\end{center}
\caption{Familiar types of branes found in Sen's weak coupling limit. The Whitney brane is the one observed in Sen's limit of an $E_8$ elliptic fibration. It can specialize into a brane-image-brane pair when $\chi$ is a perfect square and into two invariant branes on top of each other when $\chi=0$.    \label{table.brane}}
\end{table}

\subsection{Weak coupling limit and the tadpole matching condition}
 Following the point of view of  \cite{AE2}, the weak coupling limit of an elliptic fibration is a degeneration such that in the limit, the $j$-invariant becomes infinity almost everywhere. An infinite $j$-invariant corresponds to a vanishing  imaginary part of $\tau$,  which in  type IIB string theory means that the  string coupling goes to zero: $g_s\rightarrow 0$. 
 In a weak coupling limit of F-theory, the discriminant locus can decompose into multiple prime components corresponding to orientifolds and D7-branes.

In compactifications of F-theory on on an elliptically-fibered fourfold $Y\rightarrow B$, the number of  D3 branes   ($N_{D3}$) is a linear function of the Euler characteristic $ \chi(Y)$ of the elliptic fibration $Y$ and the $G_4$-flux: 
\begin{equation}
\text{D3 charge in F-theory}:\quad N_{D3}=\tfrac{1}{24}\chi(Y)-\frac{1}{2}\int_Y G_4\wedge G_4.     \label{FtheoryD3Tadpole}
\end{equation}
In a compactification of type  IIB string theory on a $\mathbb{Z}_2$ orientifold  $\rho: X\to B$ with D7-branes wrapping cycles $D_i$,  the D3 charge depends on D7-brane fluxes and the Euler characteristics of the cycles $D_i$:
\begin{equation}\label{IIBD3Tadpole}
\text{D3 charge in type IIB}:\quad 2 N_{D3}=\tfrac{1}{6}\chi(O)+\tfrac{1}{24}\sum_i\chi(D_i)+\frac{1}{2} \sum_i \int_{D_i}\mathbf{tr}(F^2_i),
\end{equation}
where $O$ is an orientifold,  $D_i$ are the surfaces wrapped by D7-branes, and $\int_{D_i} \mathbf{tr}(F^2_i)$ are fluxes localized on the D7-branes. The trace $\mathbf{tr}$ is taken in the adjoint representation. 

Duality between M-theory and type IIB suggests a matching between the  D3 charges observed in type IIB and in M-theory \ \cite{CDE}:
\begin{equation}\label{FIIB}
2\chi(Y)-24\int_Y G_4\wedge G_4=4\chi(O)+\sum_i\chi(D_i)+12 \sum_i \int_{D_i} \mathbf{tr}(F_i^2).
\end{equation} 
For vanishing $G$-fluxes and type IIB fluxes, the matching gives a  purely topological relation between Euler characteristics \cite{CDE,AE1,AE2}: 
\begin{align}\label{FIIB.Matching}
 &\text{Tadpole  matching condition}:\quad  2\chi(Y)=4\chi(O)+\sum_i\chi(D_i).
 \end{align}
In general, the  curvature contribution to the D3 tadpole in F-theory and type IIB theory do not have to match as a configuration of branes in type IIB can  recombine into a different configuration of branes with a different curvature contribution to the D3 charge but with the difference compensated by  fluxes \cite{CDE}. 

In Sen's limit, there is only one divisor $D$ which is singular and the meaning of its Euler characteristic was explained in \cite{CDE,AE1}. 
The tadpole matching condition is shown to work in a few known cases: Sen's original limit \cite{CDE,AE1}, the E$_7$, E$_6$, and D$_5$ elliptic fibrations \cite{AE1}, and
the cases of elliptic fibrations of rank one in the model introduced in \cite{EKY1}. 
In all these cases, the tadpole matching condition can be understood as a by-product of a much more general relation true at the level of the Chow group \cite{AE1,AE2}: 
\begin{align}\label{FIIB.Matching}
\text{Tadpole  matching condition}:\quad  2\varphi_* c(Y)=\rho_* \Big( {4c(O)+\sum_ic(D_i)}\Big),
 \end{align}
where $\rho: Y\to B$ is the elliptic fibration and $\rho: X\to B$ is the double cover of the base branched along  a smooth divisor that is a section of the line bundle $\mathscr{L}^{\otimes 2}$. The  relation holds for $B$ of arbitrary dimension and without imposing the Calabi--Yau condition. Mathematically, the origin of this type of relation can be understood using Verdier specialization as explained in \cite[Remark 4.5]{AE2} and in \cite{Clingher:2012rg}.

\subsection{Weak coupling limit for the SPP elliptic fibration}

In this section, we investigate the weak coupling limit of a SPP elliptic fibration and the tadpole matching condition of \cite{CDE,AE1,AE2}. 
We show that the weak coupling limit consists of an orientifold, two bi-branes wrapping the orientifold locus, and a stack of two smooth transverse invariant branes wrapping the pullback of the divisor $S=V(s)$. 

The weak coupling limit is by definition a one-parameter degeneration of the elliptic fibration such that the $j$-invariant becomes infinite almost everywhere. 
As explained in \cite{AE2}, a weak coupling limit is a semi-stable degeneration of the elliptic fibration. 
In the case of an SPP elliptic fibration, we have a fiber of type I$_2$ over the divisor $S=V(s)$. 
We recall the equation of the discriminant locus and the $j$-invariant:
\begin{equation}
\Delta= s^2(4t^3+27 s^2), \quad \quad j=-\frac{2^8t^6}{s^2(4t^3+27 s^2)}.
\end{equation}
It is clear that a weak coupling limit can be defined by 
\begin{equation}
\begin{cases}
s\to  \epsilon s, \\
t\to  h
\end{cases}
\end{equation}
which gives the following behavior at leading order in $\epsilon$: 
\begin{equation}
\Delta\to  4 \epsilon^2 s^2 h^3+O(\epsilon^4), \quad \quad j\to -\frac{2^6h^{3}}{\epsilon^2 s^2}.
\end{equation}
Following \cite{AE2}, we interpret this configuration in an orientifold theory. 
The orientifold is based on a double cover of the base branched at $h=0$. 
We start by considering the total space of the line bundle $\mathscr{L}^{\otimes 2}$ over $B$ and we introduce $\xi$ as a section of $\mathscr{L}^{\otimes 2}$. 
We then define the double cover $\rho: X\to B$  of the base $B$ branched at $V(h)$: 
\begin{equation}
X= V(\xi^2-h). 
\end{equation}

\begin{rem}[Physical interpretation]
In a type IIB theory, the orientifold involution is $\Omega (-1)^{F_L}\sigma$,
where $\sigma: X\to X: (\xi, x)\to (-\xi,x)$ is the spacetime involution, $\Omega$ is the involution on the worlsheet, and $F_L$ is the left-moving fermion number. 
The brane structure that we get describes a supersymmetric system composed of an orientifold, a USp($4$)-stack of branes, and two bi-branes on top of the orientifold. 
What is interesting here is that in contrast to the usual configuration obtained in a Sen's limit, all the branes are wrapping smooth divisors as in certain limits found in \cite{AE2}.
\end{rem}

If we pullback the discriminant to $X$, we get 
\begin{equation}
\rho^* \Delta = \xi^6  (\rho^* s)^2. 
\end{equation}
The orientifold is worth $\xi^4$ and the leftover $\xi^2$ corresponds to two bi-branes\footnote{A bi-brane is a brane-image-brane in a $\mathbb{Z}_2$ orientifold.} wrapping the orientifold.  The pullback $\rho^* s^2$ corresponds to two invariant branes forming a USp($4$)-stack of branes.

\subsection{Tadpole cancellation condition}

Comparing the D3 tadpole from type IIB and from M-theory in the absence of fluxes will impose the following tadpole matching condition: 
\begin{equation}
2\chi(Y)= 4 \chi(O)+\sum_i \chi(D_i),
\end{equation}
where $D_i$ are obtained by pulling back the divisors seen in the weak coupling limit of the discriminant locus. 

As discussed in \cite{AE2}, a more general relation is usually true: 
\begin{equation}
2\varphi_* c(Y)= 4 \rho_* c(O)+\sum_i \rho_* c(D_i).
\end{equation}
In the case of the SU($2$)$\times$U($1$)-model studied in this paper, we have 
\begin{equation}
2\chi(Y)=4 \chi(O) + \chi(O)+\chi(O)+\chi(\overline{S})+\chi(\overline{S}).
\end{equation}

\begin{thm}
Let $\varphi: Y\to B$ be an elliptic fibration defined by the crepant resolution of a SPP-model. 
Let $\rho: X\to B$ be the double cover of the base branched on   a divisor $h=0$ of class $\mathscr{L}^{\otimes 2}$. 
Then 
\begin{equation}
 \varphi_* c(Y)= 3 \rho_*  c({O})+ \rho_* c(\overline{S}).
\end{equation}
\end{thm}

\begin{proof}
The total Chern class of $X$ is
$$
c(X)= \frac{1+L}{1+2L} \rho^* c(B).
$$
The orientifold  in $X$ is ${O}=V(\xi)$, which gives by adjunction
$$
c({O})=\frac{1}{1+L} c(X)=\frac{1}{(1+2L)} \rho^* c(B).
$$
The divisor $\overline{S}=\rho^* S$ has 
\begin{equation}
c(\overline{S})=  \frac{3L}{1+3L}   c(X)=\frac{3L(1+L)}{(1+2L)(1+3L)}\rho^* c(B).
\end{equation}
Since $\rho$ is a finite morphism of degree two: 
\begin{equation}
\rho_* \rho^* c(B)= \deg(\rho) c(B)= 2 c(B).
\end{equation}
Thus 
\begin{equation}
\rho_*({O})=\frac{2L}{1+2L}c(B), \quad \rho_* c(\overline{S})= \frac{6L(1+L)}{(1+3L)(1+2L)} c(B). \end{equation}
The final result follows from the rational identity 
\begin{equation}
\frac{L}{2 L+1}+\frac{(L+1) L}{(2 L+1) (3 L+1)}=\frac{2L}{1+3L}.
\end{equation}

\end{proof}
\section*{Acknowledgements}
The authors are grateful to  Ravi Jagadeesan, Patrick Jefferson, Monica Kang, Julian Salazar, and Shu-Heng Shao for discussions. 
M.E. is supported in part by the National Science Foundation (NSF) grant DMS-1406925  and DMS-1701635 ``Elliptic Fibrations and String Theory.'' S.P. has been supported by the National Science Foundation through a Graduate Research Fellowship under grant DGE-1144152 and by the Hertz Foundation through a Harold and Ruth Newman Fellowship.

\end{document}